\documentclass[a4paper]{article}
\pdfoutput=0
\usepackage[margin=25mm,nohead]{geometry}
\usepackage{graphicx,amssymb,amsthm,amsmath,color}
\usepackage{enumerate}
\usepackage{xspace}

\usepackage[dvipdfm]{hyperref}

\newtheorem{thm}{\bfseries Theorem}
\newtheorem{theorem}{\bfseries Theorem}
        
\newtheorem{lemma}[thm]{\bfseries Lemma}

\newtheorem{corollary}[thm]{\bfseries Corollary}

\newtheorem{question}{Question}

\def\rho{\varrho}%
\newcommand{\R}{\ensuremath{\mathbb{R}}}%

\begin{document}

\title{Shortest Paths in Nearly Conservative Digraphs} 

\author{Zolt\'an Kir\'aly
\thanks{Department of Computer Science and
Egerv\'ary Research Group (MTA-ELTE),
E\"otv\"os University,
P\'azm\'any  P\'eter s\'et\'any 1/C, Budapest, Hungary.
Research was supported by 
grants (no.\ CNK 77780 and no.\ K 109240) from the National Development
 Agency of Hungary, based on a source from the Research and Technology 
Innovation Fund.
E-mail: {\tt kiraly@cs.elte.hu}}}

\maketitle

\abstract
We introduce the following notion: 
a digraph $D=(V,A)$ with arc weights $c: A\rightarrow \R$ 
is called nearly conservative if every negative cycle consists of two arcs. 
Computing shortest paths in nearly conservative digraphs is NP-hard, and even
deciding whether a digraph is nearly conservative is coNP-complete.

We show that the ``All Pairs Shortest Path'' problem is fixed parameter
tractable with various parameters for nearly conservative digraphs.
The results also apply for the special case of conservative mixed graphs.


\section{Introduction}

We are given a  digraph $D=(V,A)$, a weight (or a length) function
  $c: A\rightarrow \R$ 
is called conservative (on $D$) if no directed cycle with negative total weight
(``negative cycle'' for short) exists, and $c$ is called {\bf 
 $\mathbf\lambda$-nearly
 conservative} if every negative cycle consists of at most $\lambda$ arcs.

The APSP (All Pairs Shortest Paths) problem we are going to solve has two
parts, first we must decide whether $c$ is $\lambda$-nearly
conservative, next, if
the answer for the previous question is {\sc Yes}, then for all (ordered)
pairs $s\ne t$ of
vertices the task is to determine the length of the shortest (directed and
simple) path from $s$ to $t$. 

In this paper we concentrate on the case $\lambda=2$, a 2-nearly conservative
weight function $c$  is simply called \emph{nearly conservative} in this paper. 
A mixed graph $G=(V,E,A)$
on vertex set $V$ has the set $E$ of undirected edges and the set $A$ of
directed edges (i.e., arcs). A weight function  $c: E\cup A \rightarrow \R$ 
is called conservative if no cycle with negative total weight exists.
For a mixed graph we can associate a digraph by replacing each undirected edge
$e$ having endvertices $u$ and $v$ by two arcs $uv$ and $vu$ with weights
$c(uv)=c(vu)=c(e)$. It is an easy observation that the resulting $c$ is
nearly conservative on the resulting digraph if and only if the original
weight function was conservative on the original mixed graph, and in this case
the solution of the APSP problem remains the same.

Arkin and Papadimitriou proved in \cite{AP} that the problems of detecting
negative cycles and finding the shortest path in the absence of negative
cycles are both NP-hard in mixed graphs. Consequently, checking whether $c$ is
nearly conservative on $D$ is coNP-complete, and solving the APSP problem in
the case 
$c$ is nearly conservative on $D$ is NP-hard. In this paper we give FPT
algorithms for this problem related to various parameters.

Though it was a surprise to the author, he could not find any algorithm for
dealing with these problems (despite the fact that many paper are written
about the Chinese Postman problem on mixed graphs). We only found two more
papers that are somehow related to this topic. In \cite{GK} for the special
case of skew-symmetric graphs shortest ``regular'' paths are found in
polynomial time if no negative ``regular'' cycle exist. In \cite{BK} for the
similar special case of bidirected graphs minimum mean edge-simple cycles are
found in polynomial time, this is essentially the same as finding minimum mean
``regular'' cycles in skew-symmetric graphs. The class of nearly conservative
graphs seems to be not studied (and defined) in the literature, as well as we
could not find any FPT result about APSP.

For defining the parameters we are going to use, we first define the
notion of negative trees. Given $D$ and $c$, we associate an undirected graph 
$F=(V,E)$ as follows. Edge-set $E$ consists of pairs $u\ne v$ of vertices
for which both $uv$ and $vu$ are arcs in $A$, and $c(uv)+c(vu)<0$.
We can construct $F$ in time $O(|A|)$, and can also check whether it is a
forest. We claim that if $F$ is not a forest, then $c$ is not nearly
conservative on $D$, so our algorithm can stop with this decision.
If $F$ contains a cycle, then it corresponds to two oppositely directed cycles
of $D$, and the sum of the total weights of these two cycles are negative,
proving that $c$ is not nearly
conservative.

From now on we will suppose that $F$ is a forest, and we call its nontrivial
components (that have at least one edge) the {\bf negative trees}. 
  
Our first parameter $k_0$ is the number of negative trees, and we give an
$O(2^{k_0}\cdot n^4)$ algorithm for the APSP problem (where $n=|V|$).  Later
we refine this algorithm for parameter $k_1$, which is the maximum number of
negative trees in any strongly connected component of $D$, and finally for
parameter $k_2$, which is the maximum number of negative trees in any weakly
2-connected block of any strongly connected component of $D$ (for the
definitions see the next section). Our final algorithm also runs in time
$O(2^{k_2}\cdot n^4)$.  Consequently, if there is a constant $\gamma$ such that
every weakly 2-connected block of any strongly connected component of $D$ has
at most $\gamma$ negative trees, then we have a polynomial algorithm.

The preliminary version of this paper appeared in \cite{egres} 
for the special case of mixed graphs. In that paper  we also
gave a strongly polynomial algorithm for finding shortest exact walk (a walk
with given number of edges) in any non-conservative mixed graph.

\section{Definitions}\label{sec_defi}

For our input digraph $D$ we may assume it is simple.
An arc from $u$ to $v$ is called a \emph{loose arc} if there is another arc 
from $u$ to $v$ with a smaller weight. In a shortest path between $s$ and $t$
(if $s\ne t$) neither loops nor loose arcs can appear. Consequently, as a
preprocessing, we can safely delete these (and also keep only one copy from
multiple arcs having the same weight). 

However for our purposes multiple arcs
will be useful, so we will use them for describing the algorithm. We use
the convention that the notation $uv$ always refers to the shortest arc from $u$
to $v$. 

We call an arc $uv$ of $D$ {\bf special} if $vu$ is also an arc, and
moreover $c(uv)+c(vu)<0$. Other arcs are called \emph{ordinary}. For a special
arc $uv$ the special arc $vu$ is called its \emph{opposite}. 
As a part of the preprocessing, we add some loose arcs to $D$. For every
special arc $uv$ we add an arc $a$ from $v$ to $u$ with weight
$c(a)=-c(uv)$. By the definition of special arcs, these are really loose
arcs, as $-c(uv)>c(vu)$. We call these arcs \emph{added  ordinary arcs}, or
shortly \emph{loose arcs}.
We call the improved digraph also $D$, and its arc set is called $A$. 
Arc set $A$ is decomposed into $A=A_s\cup A_o$,
where $A_s$ is the set of special arcs, and $A_o$ is the set of ordinary
(original or added) arcs. (The main purpose of this procedure is the
following. We will sometimes work in the ordinary subdigraph $D_o=(V,A_o)$,
and we need to maintain the same reachability: if there is a path from $s$
to $t$ in $D$, then there is also a path from $s$
to $t$ in $D_o$.) Our main property remained true: if $c$ is nearly
conservative on $D$, then every negative cycle consists of two oppositely
directed special arcs. Remark: special arcs may have positive length, so
loose arcs may have negative length.
We call a path \emph{ordinary} if all its arcs are ordinary. 
Note that by the assumptions $|A|\le 2n^2$, where $n=|V|$.

Given $D$ and $c$, we associate an undirected graph $F=(V,E)$ as
follows. Edge-set $E$ consists of unordered pairs $u\ne v$ of vertices for
which $uv$ is a special arc in $A_s$. As we detailed in the Introduction, if
$F$ is not a forest, then $c$ is not nearly conservative on $D$. We consider
this process as the last phase of the preprocessing: we determine $F$, and if
it is not a forest, then we stop with the answer ``{\sc Not Nearly
  Conservative}''.

From now on we suppose that $F$ is a forest, and we call its nontrivial
components (that have at least one edge) the {\bf negative trees}. 
If $T$ is a negative tree, then ${\mathbf{V(T)}}$ denotes its vertex set, and
${\mathbf{A(T)}}$ 
denotes the set of special arcs that
correspond to its edges.
If $s,t\in V(T)$ are two vertices of $T$, then 
${\mathbf{d^T(s,t)}}$ 
denotes the length of unique path from $s$ to $t$ in $A(T)$.

A walk from $v_0$ to $v_\ell$ (or a $v_0v_\ell$-walk) is a sequence
\[W=v_0,a_1,v_1,a_2,v_2,\ldots,v_{\ell-1}, a_\ell,v_\ell\]
where $v_i\in V$ for
all $i$, and $a_j$ is an arc from $v_{j-1}$ to $v_j$ for all $j$.
A walk is closed if $v_0=v_\ell$. A closed walk is also
called here a $v_0v_0$-walk. A number $\ell$ of arcs used by a walk $W$ is
denoted by $|W|$. 
The length (or weight) $c(W)$ of a walk $W$ is defined as $\sum_{j=1}^\ell
c(a_j)$. If $W_1$ is a $s_1v$-walk and $W_2$ is a $vt_2$-walk, then
their concatenation is denoted by $W_1+W_2$. For a walk $W$  we use the notation
$W[v_i,v_j]$ for the corresponding part $v_i,a_{i+1},\ldots, a_j,v_j$ if
$i<j$.

A walk $W$ is {\bf special-simple} if no special arc is contained twice
in it, moreover, if $W$ 
contains special arc $uv$, then it does not contain its opposite $vu$.
A walk is a \emph{path} if all the vertices $v_0,\ldots,v_\ell$ are distinct.
A closed walk is a \emph{cycle} if all the vertices $v_0,\ldots,v_\ell$ are
distinct, with the exception of $v_0=v_\ell$. If $|W|=\ell=0$, then we call the
walk also an empty path (its length is 0), and in this paper unconventionally
the empty path will also be considered 
as an empty cycle. The distance $d_D(s,t)=d(s,t)$ of $t$ from $s$ is the
length of the shortest path from $s$ to $t$ (where $s,t\in V$).

The relation: there is a path in $D$ from $s$ to $t$ and also from $t$ to $s$,
is obviously an equivalence relation, its classes are called the strongly
connected components of $D$. (Notice that a negative tree always resides in
one strongly connected component.)  A weakly 2-connected block of a digraph is
a 2-connected block of the underlying undirected graph (where arcs are replaced
with undirected edges).

An algorithm is FPT for a problem with input size $n$ and parameter $k$ if
there is an absolute constant $\gamma$, and a function $f$ such that the
running time is $f(k)\cdot O(n^\gamma)$. (Originally FPT stands for ``fixed
parameter tractable'', and it is an attribute of the problem, however in the
literature usually the corresponding algorithms are also called FPT.)
In this paper we give FPT algorithms for the APSP problem for nearly
conservative digraphs.

In the simplest version we assume that there is just one negative tree and it
is spanning $V$. Next we give an algorithm for the case where we still have only
one negative tree, but it is not spanning $V$.  
These algorithms are polynomial and simple. 

Then we use various parameters:
$k_0$ is the number of negative trees in $D$,
$k_1$ is the maximum number of negative trees in
any strongly connected component of $D$, and
$k_2$ is  the maximum number of negative trees in any weakly 2-connected 
block of any strongly connected component of $D$. 
(Clearly $k_0\ge k_1\ge k_2$.)
The main goal of this paper to give an
$O(2^{k_2}\cdot n^4)$ algorithm for the APSP problem for the case $\lambda=2$,
i.e., for deciding whether $c$ is nearly conservative on $D$, and if it is,
then for calculating the distances $d_D(s,t)$ for each (ordered)
pair of vertices $s,t\in V$.

In the next section we show some lemmas. In Section \ref{veryrestr} we give
some polynomial algorithms for the case of one negative tree.  In Section
\ref{restr} we give an FPT algorithm where the parameter $k_0$ is the total
number of negative trees in $D$.  Next, in Section \ref{gen} we extend it to
the case where $k_2$ only bounds the number of negative trees in any weakly
2-connected block of any strongly connected component.

Our main goal is only giving the length of the shortest paths, in Section
\ref{paths} we detail how the actual shortest paths themselves can be
found. 

Finally in Section \ref{conclusion} we conclude the results, show their
consequences to mixed graphs, and pose some open problems. 

\section{Lemmas}\label{lemmas}

In this section we formulate some lemmas. Though each of them can be easily
proved using the newly introduced notions and the statements of the preceding
lemmas, we could not find these statements in the literature (neither in an
implicit form).

We premise some unusual aspects of nearly conservative weight functions.
Usually shortest path algorithms use the following two facts about
conservative weight functions. If $P$ is a shortest $sx$-path and $Q$ is a
shortest $xt$-path, then $P+Q$ contains an $st$-path not longer than
$c(P)+c(Q)$. If $P$ is a shortest $st$-path containing vertices $u$ and $v$
(in this order), then $P[u,v]$ is a shortest $uv$-path. These two statements
are NOT true for nearly conservative weight functions.

Remember that $D=(V,A_s\cup A_o)$ is the improved digraph
with loose arcs, and the associated graph $F$ is a forest.

\begin{lemma}\label{us}
Weight function $c$ is nearly conservative on $D$ if and only if there is no
negative special-simple closed walk. 
\end{lemma}

\begin{proof}
  If $C$ is a negative cycle consisting of at least three arcs, then it is
  also a negative special-simple closed walk. On the other hand, suppose that
  $C$ is a negative special-simple closed walk with a minimum number of arcs,
  and assume that $C$ is not a cycle, that is there are $0<i<j\le \ell$ such
  that $v_i=v_j$. Now $C$ decomposes into two special-simple closed walks with
  less arcs, clearly at least one of them has negative length, a
  contradiction.  \qed\end{proof}

\begin{lemma}\label{us2}
If $c$ is nearly conservative on $D$, and
$s,t\in V$, and $Q$ is a special-simple $st$-walk, then 
we also have an $st$-path $P$ with $c(P)\le c(Q)$, and $P$ contains only arcs
of $Q$.  
\end{lemma}

\begin{proof}
  Let $Q$ be a shortest special-simple $st$-walk (which exists by the
  previous lemma and as $c$ is nearly conservative)
having the minimum  number of arcs. 

By the previous lemma, if  $s=t$, then the empty path serves well as $P$. So we
may assume that $s\ne t$ and $Q$ is not a path, i.e., there are $0\le i<j\le
\ell$ such that $v_i=v_j$. Now $Q$ decomposes to a special-simple
$sv_i$-walk $Q_1$, a special-simple closed walk $C$ through $v_i$ and an
special-simple $v_jt$-walk $Q_2$. By the previous
lemma $C$ is nonnegative, so $c(Q_1+Q_2)\le c(Q)$, 
consequently $Q_1+Q_2$ is a not longer special-simple $st$-walk with less
number of arcs, a contradiction. 
\qed\end{proof}

Suppose $T$ is a negative tree, $u,v\in V(T)$, and $P$ is a $uv$-path in
$D'=D-A(T)$. If $c(P)< -d^T(v,u)$, then $c$ is not nearly conservative on
$D$ because otherwise $P+P^T_{vu}$ would be a negative special-simple closed
walk, where  $P^T_{vu}$ is the $vu$-path in $A(T)$. Otherwise, if 
$c(P)\ge  -d^T(v,u)$, then we have a $uv$-path $P'$ in $D'$ consisting of
loose arcs such that $c(P')\le c(P)$. Using this train of thought we get
the following lemmas that play the central role in our algorithms.

\begin{lemma}\label{kitero}
Let $T$ be a negative tree, and assume that $c$ is nearly
conservative on $D$. If $P$ is a shortest $st$-path using some vertex of
$V(T)$, then let $u$ be the first vertex of $P$ in $V(T)$, and let $v$ be the
last vertex of $P$ in $V(T)$. Then $P[u,v]$ uses only special arcs from $A(T)$. 
Consequently, if $s,t\in V(T)$, then $d(s,t)=d^T(s,t)$.
\end{lemma}

\begin{proof}
Remember that a $uv$-path in $A(T)$ may have
  positive length. Fortunately, by the definition of $u$ and $v$, there are no
  vertices of $P$ preceding $u$ or following $v$ inside $V(T)$, and this fact
  can be used successfully. 
 
  Suppose $P$ is a shortest $st$-path.  By the observation made before the
  lemma, for any $u',v'\in V(T)$, any subpath of form $P[u',v']$ that uses no
  arcs from $A(T)$ can be replaced by loose arcs without increasing the
  length. After we made all these replacements, we replaced $P[u,v]$ by a
  special-simple $uv$-walk $Q'$ such that $Q'$ contains only arcs in $A(T)$
  and loose arcs, and $c(Q')\le c(P[u,v])$.  By Lemma \ref{us2}, $Q'$
  contains a $uv$-path $P'$ with $c(P')\le c(Q')$.  We got $P'$ by eliminating
  cycles, if any cycle had positive length, then we get $c(P')<
  c(Q')$. Suppose now that all eliminated cycles had zero length, meaning that
  each one had the form $x,a,y,yx,x$, where $a$ is the loose arc from $x$ to
  $y$ and $yx$ is the special arc from $y$ to $x$. If after deleting all these
  cycles $P'$ still has a loose arc $a$ from $x$ to $y$ then it can be
  replaced safely with the special arc $xy$ yielding again a path strictly
  shorter than $Q'$.  Thus the only possibility where we can only get a $P'$
  with the same length (as $Q'$) is that the special-simple $uv$-walk $Q'$
  consisted of the $uv$-path $P^T_{uv}$ inside $A(T)$ and additionally some
  zero length cycle described above, and moreover $P'=P^T_{uv}$.  Now we claim
  that in this case the path $P[u,v]$ used only arcs from $A(T)$, i.e., it was
  also $P^T_{uv}$.  Otherwise there are vertices $x,y\in V(T)$ such that $x$
  is on $P^T_{uv}$, $y$ is not on it, and $Q'$ contains one loose arc and one
  special arc between $x$ and $y$. However in this case vertex $x$ had to
  be included twice in path $P$, a contradiction.

  To finish the proof observe that $P[s,u]+P'+P[v,t]$ is an $st$-path, and in
  the case $P[u,v]\ne P^T_{uv}$ it would be shorter than the shortest path
  $P$.  
\qed\end{proof}

\begin{lemma}\label{konz_1fa}
  Let $T$ be a negative tree, and assume that $c$ is nearly conservative on
  digraph $D'=D-A(T)$ defining distance function $d'$. Then $c$ is
  nearly conservative on $D$ if and only if for any pair of vertices $u,v\in
  V(T)$ we have $d'(u,v)\ge -d^T(v,u)$.
\end{lemma}

\begin{proof}
  We showed that the condition is necessary. Suppose that $C$ is a negative
  cycle in $D$ having at least three arcs. If it has at most one vertex in
  $V(T)$, then it is also a negative cycle in $D'$. We claim that we can
  construct a special-simple negative closed walk $C'$ which uses only
  loose arcs and arcs in $A(T)$. To achieve this goal, repeatedly take any
  subpath $C[u,v]$, where $u,v\in V(T)$, but inner vertices of $C[u,v]$ are in
  $V-V(T)$. By the condition $c(C[u,v])\ge -d^T(v,u)$, which means that
  changing $C[u,v]$ to the $uv$-path consisting of loose arcs does not
  increase the length of $C$.  We arrived at a contradiction, as the
  special-simple closed walk $C'$ contains a negative cycle which is
  impossible by the definition of loose arcs.  
\qed\end{proof}



\section{Polynomial algorithms for the case $k_0=1$}
\label{veryrestr}

First we give an $O(n^2)$ algorithm for the very restricted case, where we have
only one negative tree $T$, and moreover it spans $V$. 
We claim first that 
$c$ is nearly conservative on $D$ if and only if for each ordinary arc $uv$
we have $c(uv)\ge -d^T(v,u)$. 
If $c(uv)< -d^T(v,u)$, then we have a negative special-simple closed walk,
so $c$ is not nearly conservative by Lemma \ref{us}.
Suppose now that $c(uv)\ge -d^T(v,u)$ holds for each ordinary arc $uv$,
and $C$ is a negative cycle in $D$ with at least three arcs.
As in the proof of Lemma \ref{konz_1fa}, replace each ordinary arc $uv$ of $C$
by a $uv$-path consisting of loose arcs, this does not increase the length.
We arrive at special-simple closed walk using only special and loose arcs that
is negative. However this contradicts to the definition of loose arcs.
We also got that in this case for any pair $s,t\in V$ the
length of the shortest path is $d^T(s,t)$ by Lemma \ref{kitero}.  Consequently
it is enough to give an $O(n^2)$ algorithm for calculating distances
$d^T(s,t)$. We suppose that $V=\{1,\ldots,n\}$ and initialize a length-$n$
all-zero array $D_u$ for each vertex $u$. Then we fill up these arrays in a
top-down fashion starting from the root vertex $1$. Let $\mathcal{P}$ denote
the subset of vertices already processed, initially it is $\{1\}$. If a parent
$u$ of an unprocessed vertex $v$ is already processed, we process $v$: for
each processed vertex $x$ we set $D_v(x)=c(vu)+D_u(x)$, and set
$D_x(v)=D_x(u)+c(uv)$, and put $v$ into $\mathcal{P}$.

Next we give an $O(n^4)$ algorithm for the case where we have only one
negative tree $T$, but we do not assume it to span $V$.

In digraph $D'=D-A(T)=D_o=(V,A_o)$ using the Floyd-Warshall algorithm (see in
any lecture notes, e.g., in \cite{CLRS}), it is easy to check whether $c$ is
conservative on $D'$ in time $O(n^3)$. If it is not conservative, then we
return with output ``{\sc Not Nearly Conservative}'' (as in this case $c$
clearly cannot be nearly conservative on $D$), and if it is conservative, then
this algorithm also calculates the length $d'(s,t)$ of all shortest paths in
$D'$ (for $s,t\in V$).  If vertex $t$ is not reachable from vertex $s$, then
it gives $d'(s,t)=+\infty$ (remember that reachability is the same in $D'$ as
in $D$).

 Then we calculate the distances
$d^T(u,v)$ in time $O(n^2)$ as in the previous section. By Lemma
\ref{konz_1fa}, $c$ is nearly conservative on $D$ if and only if for all pairs
$u,v\in V_T$ we have $d'(u,v)\ge -d^T(u,v)$, this can be checked in time
$O(n^2)$. It remains to calculate the pairwise distances. If $P$ is a
shortest $st$-path, then it is either a ordinary path (having length
$d'(s,t)$), or 
it has a first arc $uu'\in A(T)$ and a last arc $v'v\in A(T)$. The
part $P[u,v]$ must reside inside $A(T)$ by Lemma \ref{kitero}.

\begin{lemma}\label{dist}
If $c$ is nearly conservative on $D$, and $T$ is the only negative tree,
then the 
distance $d(s,t)$ is
\[d(s,t)=\min\Bigl(d'(s,t), \; \min_{u,v\in
  V_T}[d'(s,u)+d^T(u,v)+d'(v,t)] \Bigr).\] 
\end{lemma}

\begin{proof}
  This is a consequence of Lemma \ref{kitero}.  The trick
  used here is that a
  shortest $su$-path and a shortest $vt$-path in $D'$ need not be
  arc-disjoint, this is the main purpose for which we introduced the notion of
  special-simple, so for the relation LHS$\le$RHS we have to use Lemma
  \ref{us2}.   
\qed\end{proof}

These values can be easily calculated for all pairs in total time $O(n^4)$,
so we are done. 

\section{FPT algorithm for parameter $k_0$}\label{restr}

In this section we suppose that there are
at most $k_0$ negative trees in $D$. 
Let $T_1,\ldots,T_{k_0}$ be the negative trees, remember that we defined
$A(T_i)$ as the set of special arcs that correspond to the edges of $T_i$. We
denote by $V_T$ the vertex set $\bigcup_i V(T_i)$.

First we compute distances  $d^{T_i}$ for
all $1\le i\le k_0$ in total time $\sum O(|V(T_i|^2)=O(n^2)$. 
Next we compute distances $d'$ in
digraph 
$D'=D-\bigcup_i A(T_i)=D_o$ in time $O(n^3)$, or stop if $c$ is not nearly
conservative on $D'$. 
 


We use dynamic programming for the calculation remained.
For all $J\subseteq \{1,\ldots,k_0\}$ we define the \emph{$J$-subproblem} as
follows. Solve the APSP problem  in digraph
$D_J=D-\bigcup\limits_{i\in \{1,\ldots,k_0\}-J} A(T_i)$, and let $d_J$ denote the
corresponding distance function if $c$ is nearly conservative on $D_J$
(otherwise, if $c$ is not nearly conservative on $D_J$ for any $J$, we stop).
We already solved the $\emptyset$-subproblem, $d_\emptyset\equiv d'$.

\begin{lemma}\label{dist2}
Suppose we solved the $(J-i)$-subproblem for every $i\in J$ and found that $c$
is nearly conservative on $D_{J-i}$. By Lemma \ref{konz_1fa}, we can check
whether $c$ is conservative on $D_J$ using only distance functions $d^{T_i}$
and $d_{J-i}$ for one element $i\in J$. If yes, then we have 
\[  d_J(s,t)=
 \min\Bigl(d_\emptyset(s,t),
\min_{i\in J}[
 \min_{u,v\in V(T_i)}(d_\emptyset(s,u)+d^{T_i}(u,v)+d_{J-i}(v,t)]\Bigr)
\]
\end{lemma}

\begin{proof}
  First we show that LHS$\ge$RHS. Let $P$ be a shortest path in $D_J$. Either
  $P$ is disjoint from $\bigcup_{j\in J}V(T_j)$, in this case its length is
$d_\emptyset(s,t)$ in graph $D_J$. The other possibility is that $P$ has some
first 
vertex $u$ in  $\bigcup_{j\in J}V(T_j)$, say $u\in V(T_i)$. Let $v$ denote the
last vertex of $P$ in $V(T_i)$. That is, $P[s,u]$ goes inside $D_\emptyset$ 
and $P[v,t]$ goes inside $D_{J-i}$, and, by Lemma \ref{kitero},
$P[u,v]$ goes inside $A(T_i)$.

To show that LHS$\le$RHS we only need to observe that if
$P_1$ is an $su$-path in $D_\emptyset$, $P_2$ is a $uv$-path in $A(T_i)$, and
$P_3$ is a $vt$-path in $D_{J-i}$, then $P_1+P_2+P_3$ is a special-simple
$st$-walk. 
\end{proof}

Remember that 'solving the APSP problem' is defined in this paper as first
checking nearly conservativeness, and if $c$ is nearly conservative, then
calculate all shortest paths. As solving one subproblem needs $O(n^4)$ steps,
we proved the following 

\begin{theorem}\label{thm_restr}
If $D$ has $k_0$ negative trees, then the dynamic programming algorithm
given in this section correctly solves the APSP problem in time $O(2^{k_0}\cdot
n^4)$. 
\end{theorem}

The weak blocks of a digraph refer to the 2-connected blocks of the underlying
undirected graph. It is well known that the block-tree of an undirected graph
can be determined in time $O(n^2)$ by DFS. If we have this decomposition and
we also calculated APSP inside every weak block, then we can also calculate
APSP for the whole digraph in additional time $O(n^3)$. Consequently we have

\begin{corollary}\label{cor_weakblock}
If every weak block of $D$ contains at most $k_0'$ negative trees, then we can
solve the APSP problem in time $O(2^{k_0'}\cdot n^4)$. 
\end{corollary}

\section{General FPT algorithm for parameters $k_1$ and
 $k_2$}\label{gen}

Suppose every strongly connected component of $D$ contains at most $k_1$
negative trees.
By the previous section we can solve the  APSP problem inside each strongly
connected component in total time $O(2^{k_1}\cdot n^4)$. If for any of them we
found that $c$ is not nearly conservative, then we stop and report the fact
that $c$ is not nearly conservative on $D$. Henceforth in this section we
assume that 
for every strongly connected component $K$ of $D$,
$c$ is nearly conservative on $K$. (In this situation clearly $c$ is nearly
conservative on $D$.)
The distance function
restricted to component $K$  is denoted by $d_K$. If $s,t\in V(K)$, then every
$st$-path goes inside $K$, thus $d(s,t)=d_K(s,t)$.
It remains to calculate
APSP in $D$ for pairs $s,t$, that are in different strongly connected
components.  

We construct a new acyclic digraph $D^*$ by first substituting every
strongly connected component $K$ by acyclic digraph $D^*_K$ as
follows. Suppose $V(K)=\{x_1^K, x_2^K,\ldots, x_r^K\}$, the vertex set of
$D^*_K$ will consist of $2r$ vertices, $\{a_1^K, a_2^K,\ldots, a_r^K,$
$b_1^K, b_2^K,\ldots, b_r^K\}$. For each $1\le i,j\le r$ the digraph $D^*_K$ 
contains arc $a_i^Kb_j^K$ with length $d_K(x_i^K,x_j^K)$.

In order to finish the construction of $D^*$, for every arc $x_i^Kx_j^L$ of
$D$ connecting two different strongly connected components $K\ne L$, digraph
$D^*$ contains the arc $b_i^Ka_j^L$ with length $c(x_i^Kx_j^L)$. It is easy
to see that $D^*$ is truly acyclic and has $2n$ vertices. As $D^*$ is a
simple digraph, paths can be given by only listing the sequence of its
vertices.  We can calculate APSP in $D^*$ in time $O(n^3)$ by the method of
Mor\'avek \cite{M} (see also in \cite{CLRS}) if we run this famous algorithm
from all possible sources $s$.  It gives distance function $d_{D^*}$ (where
if $t$ is not reachable from $s$, then we write $d_{D^*}
(s,t)=+\infty$). The total running time is still $O(2^{k_1}\cdot n^4)$. We
remark that if every strongly connected component  has
a spanning negative tree, then the running time is $O(n^3)$.

\begin{theorem}\label{thm_gen}
Suppose $s=x_{i_0}^{K_0}\in V(K_0)$ and $t=x_{j_r}^{K_r}\in V(K_r)$ where
$K_0\ne K_r$ are different strongly connected components of $D$. Then the
shortest $st$-path in $D$ has length exactly
$d_{D^*}(a_{i_0}^{K_0},b_{j_r}^{K_r})$. 
\end{theorem}

\begin{proof}
  Vertex $t$ is not reachable from $s$ in $D$ if and only if $b_{j_r}^{K_r}$
  is not reachable from $a_{i_0}^{K_0}$ in $D^*$. Otherwise, suppose
  that $a_{i_0}^{K_0},b_{j_0}^{K_0},
  a_{i_1}^{K_1},b_{j_1}^{K_1},\ldots,a_{i_r}^{K_r},b_{j_r}^{K_r}$ is a
  shortest path $P$ in $D^*$. For $0\le \ell\le r$ let path $P_\ell$ be a
  shortest path in $D$ from $x_{i_\ell}^{K_\ell}$ to $x_{j_\ell}^{K_\ell}$,
  this path obviously goes inside $K_\ell$. We can construct an $st$-path $Q$ in
  $D$ with the same length as $P$ has in $D^*$: $Q=P_0+
  x_{j_0}^{K_0}x_{i_1}^{K_1}+ P_1+ x_{j_1}^{K_1}x_{i_2}^{K_2}+ P_2+ \ldots +
  P_{r-1}+ x_{j_{r-1}}^{K_{r-1}}x_{i_r}^{K_r}+ P_r$.

  For the other direction, suppose that there are strongly connected
  components $K_0, K_1, \ldots, K_r$, such that the shortest $st$-path $Q$ in
  $D$ meets these components in this order, and for all $\ell$ the path $Q$
  arrives into $K_\ell$ at vertex $x_{i_\ell}^{K_\ell}$ and leaves $K_\ell$ at
  vertex $x_{j_\ell}^{K_\ell}$. As $Q$ is a shortest path it clearly contains
  a path of length $d_{K_\ell}(x_i{_\ell}^{K_\ell},x_{j_\ell}^{K_\ell})$
  inside $K_\ell$ for each $\ell$, consequently the following path has the
  same length in $D^*$: $P= a_{i_0}^{K_0},b_{j_0}^{K_0},
  a_{i_1}^{K_1},b_{j_1}^{K_1},\ldots, a_{i_r}^{K_r},b_{j_r}^{K_r}$.
  \qed\end{proof}

Using Corollary \ref{cor_weakblock} we easily get the following more general
statements.

\begin{corollary}\label{cor_weakblock_gen}
If every weak block of any strongly connected component of $D$ contains at
most $k_2$ negative trees, then we can solve the APSP problem in time
$O(2^{k_2}\cdot n^4)$. 
\end{corollary}

\begin{corollary}
  If there is an absolute constant $\gamma$, such that in any weak block of
  any strongly connected component of $D$ there are at most $\gamma$ negative
  trees, then there is a polynomial time algorithm for the APSP problem that
  runs in time $O_\gamma(n^4)$. 
\end{corollary}

\section{Finding the paths}\label{paths}

In this section we assume that $c$ is nearly conservative on $D$.

We usually are not only interested in the lengths of the shortest paths, but
also some (implicit) representation of the paths themselves. The requirement
for this representation is that for any given $s$ and $t$, one shortest
$st$-path $P$ must be computable from it in time $O(\ell)$ if $\ell$ is the
number of arcs in $P$.  

It is well known (see e.g., in \cite{CLRS}) that both the algorithm of Floyd
and Warshall and the algorithm of Mor\'avek  can compute predecessor matrices
$\Pi$ (by increasing the running time by a constant factor only), with the
property that for each $s\ne t$ the entry $\Pi(s,t)$ points to the
last-but-one vertex of a shortest $st$-path. This representation clearly
satisfies the requirement described in the previous paragraph.  

For a digraph $H$ let $\Pi_H$ denote the predecessor matrix of this type, and
suppose that for each strongly connected component $K$ we computed
$\Pi_K$, and we also computed $\Pi_{D^*}$. Then $\Pi_D$ is easily
computable as follows. Suppose that $s=x_{i_0}^{K_0}$ and $t=x_{j_r}^{K_r}$,
and $\Pi_{D^*}(a_{i_0}^{K_0},b_{j_r}^{K_r})=a_{i_r}^{K_r}$. If $i_r\ne j_r$, then
define $\Pi_D(s,t)=\Pi_{K_r}(x_{i_r}^{K_r},x_{j_r}^{K_r})$, otherwise let
$b_{j_{r-1}}^{K_{r-1}}=\Pi_{D^*}(a_{i_0}^{K_0},a_{i_r}^{K_r})$ and define
$\Pi_D(s,t)=x_{j_{r-1}}^{K_{r-1}}$. 

It remained to compute the predecessor matrices $\Pi_K$ in the case where $K$
is a strongly connected component of $D$. In accordance with Section \ref{restr}
from now on we call $K$ as $D$ (and forget the other vertices of the digraph),
and the matrix we are going to determine is simply $\Pi$. 

If $s$ and $t$ are vertices of the same negative tree $T_i$, then the method
given in the first paragraph of Section \ref{veryrestr} easily calculates
$\Pi(s,t)=\Pi_{A(T_i)}(s,t)$. Next we
call the Floyd-Warshall algorithm on $D'$, and it
can give $\Pi_{D'}$, then during the dynamic programming algorithm 
we determine matrices
$\Pi_{D_J}$ for all $J$. 

Given $s$ and $t$, 
by Lemma \ref{dist2} if the minimum is $d_\emptyset(s,t)$, then 
$\Pi_{D_J}(s,t)=\Pi_{D_\emptyset}(s,t)$, otherwise we find $i,u,v$ giving the
minimum value. If $v\ne t$, then $\Pi_{D_J}(s,t)=\Pi_{D_{J-i}}(s,t)$,
otherwise if $v=t$ but $u\ne v$, then $\Pi_{D_J}(s,t)=\Pi_{A(T_i)}(s,t)$,
and finally if $v=t=u\ne s$ then $\Pi_{D_J}(s,t)=\Pi_{D_\emptyset}(s,t)$.

Extending this setup for weak blocks is obvious.

\section{Conclusion and open problems}\label{conclusion}

We gave FPT algorithms for the NP-hard 
APSP problem in nearly conservative graphs
regarding with various parameters.  

For mixed graphs we have the following consequence. As nonnegative undirected
edges can be replaced by two opposite arcs, we may assume that every
undirected edge has negative length. Here the negative trees are the
nontrivial components made up by undirected edges, and APSP problem is to
check whether $c$ is conservative on a mixed graph $G$, and if {\sc
Yes}, then calculate the pairwise distances.

Remember, that for mixed graphs the APSP problem contains checking
conservativeness, and if $c$ is conservative on the mixed graph, then all
shortest paths should be calculated.

\begin{corollary}\label{mixed}
If every weak block of any strongly connected component of a mixed graph 
contains at
most $k_2$ negative trees, then we can solve the APSP problem in time
$O(2^{k_2}\cdot n^4)$. 
\end{corollary}

Finally we pose three open problems. A weight function is
even-nearly conservative if every negative cycle consist of an even number of
arcs. 

\begin{question}
  Is there an FPT algorithm for shortest paths if $c$ is  3-nearly
conservative? (The parameter should not contain the number of negative
triangles.)
\end{question}

\begin{question}
 Is there a polynomial or FPT algorithm for recognizing  even-nearly
 conservative  weights? This would be interesting even if we restrict the
 digraph to be symmetric (i.e., every arc has its opposite). 
\end{question}

\begin{question}
  Is there an FPT algorithm for shortest paths if $c$ is
  $\lambda$-nearly 
conservative, using some parameter $k$ of ``inconvenient components''
(should be defined accordingly) and also
$\lambda$? 
\end{question}

\section*{Acknowledgment}

The author is thankful  to Andr\'as Frank who asked a special case of this
problem, and also to D\'aniel Marx who proposed the generalization to nearly
conservative digraphs.


\begin{thebibliography}{99}

\bibitem{AP}
{\scshape E.M.~Arkin, C.H.~Papadimitriou}
\newblock {On negative cycles in mixed graphs}, 
\newblock  {\itshape Operations Research Letters}
{\bf 4} (3) (1985), pp.~113--116.

\bibitem{BK}
{\scshape M.A.~Babenko, A.V.~Karzanov}
\newblock {Minimum mean cycle problem in bidirected and
skew-symmetric graphs}, 
\newblock  {\itshape Discrete Optimization}
{\bf 6} (2009), pp.~92--97.

\bibitem{CLRS}
{\scshape T.H.~Cormen, C.E.~Leiserson, R.L.~Rivest, C.~Stein}
\newblock {Introduction to Algorithms}, 
\newblock  {\itshape MIT Press, Cambridge}
third edition, (2009)

\bibitem{GK}
{\scshape A.V.~Goldberg, A.V.~Karzanov}
\newblock {Path problems in skew-symmetric graphs}, 
\newblock  {\itshape Combinatorica}
{\bf 16} (3) (1996), pp.~353--382.

\bibitem{K}
{\scshape R.M.~Karp}
\newblock {A characterization of the minimum cycle mean in a digraph}, 
\newblock  {\itshape Discrete Mathematics}
{\bf 23} (1978), pp.~309--311.

\bibitem{egres}
{\scshape Z.~Kir\'aly}
\newblock {Shortest paths in mixed graphs}, 
\newblock {\itshape Egres Technical Report TR-2012-20},
\href{http://www.cs.elte.hu/egres/www/tr-12-20.html}
{www.cs.elte.hu/egres/}

\bibitem{M}
{\scshape J.~Mor\'avek}
\newblock {A note upon minimal path problem}, 
\newblock  {\itshape Journal of Mathematical
Analysis and Applications}
{\bf 30} (1970), pp.~702--717.


\end{thebibliography}
\end{document}